\documentclass[a4paper,10pt]{article}

\usepackage{amsthm}  
\usepackage{amsmath} 
\usepackage{amssymb}
\usepackage{bm} 
\usepackage{relsize} 
\usepackage{algpseudocode} 
\usepackage{tikz} 
\usepackage{verbatim}
\usepackage[colorlinks = true]{hyperref}
\usetikzlibrary{arrows}
\usepackage{authblk}

\theoremstyle{plain}
\newtheorem{theorem}{Theorem}[section]
\newtheorem{lemma}[theorem]{Lemma}

\newtheorem{proposition}[theorem]{Proposition}
\theoremstyle{definition}

\DeclareMathOperator*{\bigsqcap}{\rotatebox[origin=c]{180}{$\bigsqcup$} } 

\DeclareMathOperator*{\lant{lant}}

\begin{document}

\title{Description logics as polyadic modal logics}

\author[*]{Jonne Iso-Tuisku}
\author[**]{Antti Kuusisto}
\affil[*]{University of Helsinki}
\affil[**]{University of Helsinki and Tampere University}

\date{}

\maketitle

\begin{abstract}
\noindent
We study extensions of standard description logics to the framework of polyadic modal logic. We promote a natural approach to such logics via general relation algebras that can be used to define operations on relations of all arities. As a concrete system to illustrate our approach, we investigate the polyadic version of $\mathcal{ALC}$ extended with relational permutation operators and tuple counting. The focus of the paper is conceptual rather than technical.
\end{abstract}

\section{Introduction}

Description logics can be naturally represented as algebras of different kinds. A natural basic setting for this can be constructed from \emph{relation operators}, as 
defined in \cite{alg2021}. Basically, a relation operator is a function that maps tuples of relations into relations. More rigorously, a relation operator of arity $k$ is a class that 
outputs, when given an set $A$ as an input, a function
\[f^A : (\mathit{rel}(A))^k \rightarrow \mathit{rel}(A)\]
where $\mathit{rel}(A)$ is the set of all relations over
the set $A$. Two further constraints hold:
\begin{enumerate}
    \item 
The relations in $\mathit{rel}(A)$ are 
\emph{arity definite}, meaning that 
empty relations have an associated arity; we denote the $n$-ary
empty relation by $\emptyset_n$. The reason for considering arity definite relations is that the complement
relation of $\emptyset_n$ in $\mathit{rel}(A)$ is then the $n$-ary 
total relation $A^n$. This ensures, e.g., that the complement of the complement of $A^n$ is $A^n$ itself for every $n$.
\item
The relation operator $f$ is isomorphism invariant in the sense that if
the structures
$(A,S_1,\dots , S_k)$ and $(B,T_1,\dots , T_k)$ are isomorphic via a bijection $g$, then so
are $(A, f^A(S_1,\dots , S_k))$ and $(B, f^B(T_1,\dots , T_k))$ (also 
via $g$). Note here that the isomorphisms make sure that the arities of corresponding relations (e.g., the arities of $S_i$ and $T_i$) match, even in the case that the relations are empty.
\end{enumerate}

Any sequence $(f_i)_{i\in I}$ of
relation operators defines an algebraic system denoted $\mathrm{GRA}(f_i)_{i\in I}$, where
$\mathrm{GRA}$ stands for \emph{general relation algebra} and $I$ is an arbitrary set of indices. When $(f_i)_{i\in I}$ is a finite
sequence $(f_1,\dots, f_n)$, then we simply write $\mathrm{GRA}(f_1,\dots , f_n)$.

Now consider an arbitary relational structure $\mathfrak{A}$. The 
system $\mathrm{GRA}(f_i)_{i\in I}$ defines an algebra over $\mathfrak{A}$ in a natural way.
The terms of the algebra are built as follows.

\begin{enumerate}
    \item 
    Every relation symbol $Q$ in the vocabulary of
    $\mathfrak{A}$ is an atomic term interpreted as
    the relation $Q^{\mathfrak{A}}$.
    %
    %
    %
    \item
    If $t_1,\dots , t_n$ are terms and $f$ an $n$-ary 
    relation operator,\footnote{It is safe not to differentiate between
    a relation operator and the symbol $f$ denoting it, as this causes no real problems in our study}
    then $f(t_1,\dots , t_n)$ is a term 
    interpreted in the obvious way as the relation obtained by applying $f^A$ to
    the interpretations of $t_1,\dots , t_n$. Here $A$ is of course the domain of $\mathfrak{A}$.
\end{enumerate}

The algebra defined by $\mathrm{GRA}(f_i)_{i\in I}$ over a model $\mathfrak{A}$ gives rise to a set of relations, the set of those relations that the terms of $\mathrm{GRA}(f_i)_{i\in I}$ define over $\mathfrak{A}$. We note the obvious fact that for ecample first-order logic FO similarly defines a set of relations over $\mathfrak{A}$: a formula $\varphi(x_1,\dots, x_k)$ in the free variables $x_1,\dots , x_k$
defines the $k$-ary relation 
\[ \{\, (a_1,\dots , a_k)\in A^k\, |\, \mathfrak{A}\models \varphi(a_1,\dots ,a_k )\ \}. \]
Note that here we assume that the variables of FO consist of
the set $\{ x_i\ |\ i \in \mathbb{Z}_+ \}$, and we use the ordering of the subindices of the variables to define the relation that $\varphi(x_1,\dots , x_k)$ corresponds to. Thus for example $\psi(x_2,x_1)$ defines the inverse of the relation defined by $\varphi(x_1,x_2)$. Note also that the variables themselves do not matter, only their ordering, so $\varphi(x_1,\dots , x_k)$ and $\varphi(x_{i_1},\dots , x_{i_k})$
for $i_1 < \dots < i_k$ define the same relation.
We say that a logic L is \emph{equiexpressive}
with $\mathrm{GRA}(f_i)_{i\in I}$ if L and $\mathrm{GRA}(f_i)_{i\in I}$ can define precisely the same relations over all relational models $\mathfrak{A}$.

Now consider the following seven relation operators.

\begin{enumerate}
    \item
    $p$ denotes the \emph{cyclic permutation} operator interpreted so that if $R\subseteq A^k$ is a $k$-ary 
    relation over $A$, then \[p^A(R) =
    \{(a_k,a_1,\dots , a_{k-1})\ |\ (a_1,\dots , a_k)\in R\},\]
    that is, $p^A$ places the last element of every tuple to the beginning of that tuple. This operator leaves relations of arity less than two as they are.
    \item
    $s$ denotes the \emph{swap} operator,
    \[s^A(R) = \{(a_1,\dots , a_{k-2}, a_k, a_{k-1})\ |\ (a_1,\dots , a_k)\in R\}\]
    which swaps the last two elements of each tuple. Like $p$, this operator does nothing to relations of arity less than two.
    \item
    $I$ denotes the \emph{identification} operator, 
    \[I^A(R) = \{(a_1,\dots , a_{k-1})\ |\ (a_1,\dots , a_k)\in R\text{ and } a_{k-1} = a_k\}.\]
    Like $p$ and $s$, this operator leaves relations of arity less that two as they are.
Informally, $I$ scans all tuples and accepts those where the last two elements are the same, and in 
such tuples, the last repetition is then deleted. Thus $I$ decreases the arity of input relations.
    \item
    $\neg$ is the complementation operator, so $\neg^A$ returns the complement $A^k \setminus R$ of the 
    $k$-ary relation $R$. Note that if $R = \emptyset^k$, then $\neg^A(R)$ is $A^k$.
    \item
    $J$ is the \emph{join} operator, or \emph{Cartesian product} operator. If $R$ and $S$ are $k$-ary and $n$-ary relations over $A$, then $J^A(R,S)$ is the $(k + n)$-ary relation 
    \[       \{ (a_1,\dots , a_k,b_1,\dots , b_n)\ |\ 
    (a_1,\dots , a_k) \in R\text{ and } (b_1,\dots , b_n)\in S\}.\]
    \item
    $\exists$ is the \emph{existential quantification} operator or
    \emph{projection} operator, 
    \[\exists^A(R) = \{(a_1,\dots , a_{k-1})\ |\ (a_1,\dots , a_k)\in R\text{ for some }a_k\in A \}.\]
    \item The operator $e$ denotes the constant operation interpreted as  
    the \emph{equality} relation $\{(a,a)\ |\ a\in A\}$ on
    every set $A$.
\end{enumerate}

The following is proved in \cite{alg2021}.

\begin{proposition}
$\mathrm{GRA}(e,p,s,I,\neg,J,\exists)$ and
$\mathrm{FO}$ are equiexpressive.
\end{proposition}

Now, a relational FO-atom is any formula of 
the form $R(y,\dots , z)$ where $R$ is a $k$-ary relation symbol; the equality symbol is not considered a relation symbol here.
The following result from \cite{alg2021} is also interesting to us.

\begin{proposition}\label{atomproposition}
$\mathrm{GRA}(p,s,I)$ is equiexpressive with the set of relational $\mathrm{FO}$-atoms.
$\mathrm{GRA}(e,p,s,I)$ is equiexpressive with the set of $\mathrm{FO}$-atoms.
\end{proposition}

Description logics can naturally be conceived as general relation algebras of different kinds. The related perspectives can have a particularly unifying flavour when considering systems with $k$-ary relations in addition to the usual binary ones. Coupling general relation algebras with polyadic modal logic is especially interesting.

To illustrate this, consider first the following syntax.
\begin{align*}
C &::=  \top \mid \bot \mid A \mid  \neg C  \mid (C_1 \sqcap C_2) \mid \exists R . (C_2,\ldots,C_{n})
\end{align*}
where the symbol $A$ is chosen from the set of \emph{atomic concepts} and the $k$-ary symbol $R$ from the set of \emph{atomic roles}. We call this system $\mathcal{ALCP}$, or \emph{polyadic} $\mathcal{ALC}$. 
The semantics is the usual one for polyadic modal logic, so
the interpretation of $\exists R . (C_2,\ldots,C_{n})$ consists of 
the elements $u$ in the domain such that $R(u,v_2,\dots v_n)$ for
some $v_2,\dots , v_n$ that belong to the interpretations of $C_2,\dots , C_n$, respectively.

Now, recall that $\mathcal{ALCI}$ is the extension of $\mathcal{ALC}$ with inverse roles, so it is natural to ask what the corresponding extension of $\mathcal{ALCP}$ would be. Following the approach of \cite{DBLP:conf/dlog/Kuusisto16}, one could argue that the corresponding extension is obtained by adding the capacity to define all the possible permutations of higher-arity roles to the picture. This can be naturally achieved by adding the cyclic permutation operator $p$ and the swap operator $s$ to the setting. Let us do this formally.

Let the syntax of $\mathcal{ALCP}(p,s)$ be obtained via the grammar
\begin{align*}
\mathcal{R} &::=  R \mid \mathcal{R}^p \mid \mathcal{R}^s\\
C &::=  \top \mid \bot \mid A \mid  \neg C  \mid (C_1 \sqcap C_2) \mid \exists \mathcal{R} . (C_2,\ldots,C_{n})
\end{align*}
where again $R$ and $A$ are atomic symbols.
The term $\mathcal{R}^p$ denotes the relation that can be obtained by applying the cyclic permutation operator $p$ to the relation corresponding to $\mathcal{R}$, and similarly, $\mathcal{R}^s$ is obtained from $\mathcal{R}$ by the swap operator.

Now, all permutations of a tuple can be obtained via the use of $p$ and $s$ (this follows from well known results on permutations and is also proved in \cite{alg2021}). Thereby $\mathcal{ALCP}(p,s)$ can be seen as a natural extension of $\mathcal{ALCI}$ to contexts with relations of arbitrary arities. It is the natural polyadic extension of $\mathcal{ALCI}$. In the light of Proposition \ref{atomproposition}, also for example $\mathcal{ALC}(p,s,I)$ is a rather natural system, with $p,s$ and $I$ being able to define precisely all roles obtainable by a relational first-order atom. The syntax of $\mathcal{ALCP}(p,s,I)$ extends that of $\mathcal{ALCP}(p,s)$ by the role constructor $I$, with $\mathcal{R}^I$ interpreted as the role consisting of the tuples $(a_1,\dots , a_{k-1})$ such that $(a_1,\dots , a_k)$ belongs to the interpretation of $\mathcal{R}$ and 
we have $a_{k-1} = a_k$. The yet further extension $\mathcal{ALCP}(e,p,s,I)$ is also interesting in
the light of Proposition \ref{atomproposition}, with the binary constant role $e$ enabling also the definition of
equality atoms.

There is a wide range of description logics that can be built in this spirit, and all this calls for further investigation. Firstly, we can consider polyadic logics of the form $\mathcal{ALCP}(f_i)_{i\in I}$ for any 
sequence $(f_i)_{i\in I}$ of relation operators. Studying systems that are fragments of FO is interesting, and also going beyond first-order logic is of course natural, e.g., by investigating different generalized quantifiers, defining generalizations of the transitive closure operator to higher-arity contexts, et cetera. Secondly,
many systems $\mathrm{GRA}(f_i)_{i\in I}$ can be regarded as interesting description logics as they stand, without involving polyadic modal logic. Indeed, the approach via polyadic modal logic is only one of many 
related possibilities.

In this paper we take some steps towards better understanding description logics based on polyadic modal logic, thereby promoting the setting of polyadic modal logic---coupled with general relation algebras---as a worthy framework for related work. As a particular technical demonstration, we show that
the concept satisfiability problem for $\mathcal{ALCQP}(p,s)$ is \textsc{PSpace}-complete.
Here $\mathcal{Q}$ denotes counting of tuples; see the preliminaries section for the details.
The result concerning $\mathcal{ALCQP}(p,s)$ is also a new result concerning a new fragment\footnote{Of course $\mathcal{ALCQP}(p,s)$ nevertheless
relates quite directly to, e.g., $\mathcal{DLR}$.} of first-order logic, but nevertheless, the result in
itself is neither surprizing nor particularly difficult to obtain. The point in this article is more to promote the polyadic framework, not so much to prove yet another little technical result. \emph{We believe the setting can quite directly and naturally generalize an extremely wide range of results in standard description logics to the context of higher order relations, typically without changing the complexity class.} This is perhaps the main metaprinciple promoted by this paper.

As a further concrete result, we also provide a very simple algebraic
characterization for $\mathcal{ALC}$. The point is to illustrate a more
direct use of general relation algebras,  
this time without a detour via polyadic modal logics.

We note that using polyadic modal logic as a basis for description logics has been previously studied in some detail in \cite{DBLP:conf/dlog/Kuusisto16}, which also makes use of algebraic operators on higher-arity relations. The current paper builds upon that study.

\section{Preliminaries}

We define the description logic $\mathcal{ALCQP}$ by the grammar
\begin{align*}
%
%
        C &::=  \top \mid \bot \mid A \mid  \neg C  \mid (C_1 \sqcap C_2) \mid {\geq} k R . (C_2,\ldots,C_{n}),
\end{align*}
where the following conditions hold.
\begin{enumerate}
\item
$\top$, $\bot$ are atomic constant concepts.
\item
$A$ is an atomic concept (concept name).
\item
$R$ is an atomic role (role name) associated with an arity at least $2$.
\item
The arity of the role $R$ in the concept ${\geq} k R . (C_2,\ldots,C_{n})$ is $n$,
and $k$ is a positive integer encoded in binary. 
\end{enumerate}
The description logic $\mathcal{ALCQP}(p,s)$ extends $\mathcal{ALCQP}$ and is defined by the grammar
\begin{align*}
        \mathcal{R} &::= R \mid \mathcal{R}^s \mid \mathcal{R}^p  \\
        C &::=  \top \mid \bot \mid A \mid  \neg C  \mid (C_1 \sqcap C_2) \mid {\geq} k \mathcal{R} . (C_2,\ldots,C_{n})
\end{align*}
where $\mathcal{R}$ is a (possibly non-atomic) role of arity at least $2$.

An \textit{interpretation} $\mathcal{I} = (\Delta^\mathcal{I},\cdot^\mathcal{I})$ is a pair that consists  
of a non-empty set $\Delta^\mathcal{I}$ called the 
\textit{domain} of $\mathcal{I}$, and a map $\cdot^\mathcal{I}$ called
a \textit{valuation}.  
The valuation maps each role name $R$ of arity $n$ to a relation $R^\mathcal{I} 
\subseteq (\Delta^\mathcal{I})^n$ and each concept
name $A$ to a set $A^\mathcal{I} \subseteq \Delta^\mathcal{I}$. The constant concepts are interpreted in the usual way, $\top^\mathcal{I} := \Delta^\mathcal{I}$ and $\bot^\mathcal{I} := \emptyset$.
Below we may write $C^\mathcal{I}u$ to indicate that $u \in C^\mathcal{I}$.
Likewise, $\mathcal{R}^\mathcal{I}uv$ means that
$(u,v) \in \mathcal{R}^\mathcal{I}$.
For denoting interpretations, we mostly use symbols $\mathcal{I}$, $\mathcal{I'}$ and $\mathcal{J}$.

The (composed) roles $\mathcal{R}$ are interpreted as follows.
\begin{itemize}
%
%
%
\item $(\mathcal{R}^s)^\mathcal{I} := \{ (u_1, \ldots,u_{n-2},u_{n},u_{n-1} ) \mid
( u_1,\ldots,u_n ) \in \mathcal{R}^\mathcal{I} \}$. The operator $s$ is called the \textit{swap} operator, or
\emph{swap permutation} operator. It swaps (only) the last two elements of each tuple. Over binary relations, $s$ is equivalent to the standard inverse 
operator that produces the inverse of a binary relation.
\item $(\mathcal{R}^p)^\mathcal{I} := \{ ( u_n,u_1,\ldots,u_{n-1}) \mid
( u_1,\ldots,u_n ) \in \mathcal{R}^\mathcal{I} \}$. The operator $p$ is called the \textit{cyclic permutation} (or \textit{circular shift}) operator. It moves the last element of each tuple to the first position.
As $s$, the operator $p$ is equivalent to the standard inverse operator on binary relations.
\end{itemize}

%

It is well known that by circular shift and swap permutations, an arbitrary finite tuple can be permuted in an arbitrary way. This follows from basic results in algebra. For an easy proof in the context of general relation algebras, see \cite{alg2021}.
Thus, combining the swap and circular shift operators in different ways, e.g. $((R^s)^{p})^{p}$, we obtain all permutations of $R$. 
We shall often lighten our notation and omit parentheses, e.g., for $((R^s)^{p})^{p}$, we write $R^{s\, p\, p}$. 
Furthermore, we often denote strings like ${s\, p\, p}$ by $\pi$.
Such strings composable from $p$ and $s$ are called \emph{permutation strings}. 
We note that also the empty string of operators $p$ and $s$ is of course a permutation string.



%
%
Negated concepts and conjunctions of concepts are interpreted in the usual way:
\begin{itemize}
\item $(\neg C)^\mathcal{I} := \Delta^\mathcal{I} \setminus C^\mathcal{I}$
\item $(C \sqcap D)^\mathcal{I} := C^\mathcal{I} \cap D^\mathcal{I}$ 
\end{itemize}
The \textit{cardinality restrictions} $\geq k$
are interpreted as follows.
\begin{itemize}
\item $(\, {\geq} k \mathcal{R}. (C_2,\ldots, C_{n}) \,)^\mathcal{I} := \\ {\{ x \mid   {|\{} \bar{v} \mid \bar{v} = (x,u_2,\ldots,u_{n}) \in \mathcal{R}^\mathcal{I} \text{ and } C_i^\mathcal{I}u_i \text{ for each i} \} | \geq k \}}.$
\end{itemize}

We may use the following shorthands:
\begin{itemize}
\item
$\exists \mathcal{R}.(C_2,\ldots,C_{n})$ is shorthand for ${\geq} 1 \mathcal{R}.(C_2,\ldots,C_{n})$.
\item
${<} k \mathcal{R}.(C_2,\ldots,C_{n})$ is shorthand for $\neg {\geq} k \mathcal{R}.(C_2,\ldots,C_{n})$.
\item
$\forall \mathcal{R}.(C_2,\ldots,C_{n})$ is shorthand for ${< }1 \mathcal{R}.(\neg C_2,\ldots, \neg C_{n})$.
This is analogous to the definition of the polyadic box operator, and it is
equivalent to $\neg\exists \mathcal{R}.(\neg C_2,\ldots, \neg C_{n})$.
\item
${=}k\mathcal{R}.(C_2,\ldots,C_{n})$ is shorthand for 
$$ {\geq}k\mathcal{R}.(C_2,\ldots,C_{n})  \sqcap  {<}(k+1) \mathcal{R}.(C_2,\ldots,C_{n}).$$
\end{itemize}

For a concept $C$ and $u \in \Delta^{\mathcal{I}}$, we
may write $\mathcal{I},u\models C$ to denote that $u\in C^{\mathcal{I}}$.

\subsection{Concept satisfiability} \label{csat_empty_T_ss}

A concept  $C$ is satisfiable if
there exists an interpretation $\mathcal{I}$ such that $C^\mathcal{I} \neq \emptyset$.
In other words, there is an element $a$ in the domain $\Delta^\mathcal{I}$ such that
$a \in C^\mathcal{I}$.
In this subsection, we show that 
concept satisfiability for $\mathcal{ALCQP}(p,s)$ is \textsc{PSpace}-complete
by reducing
that problem to the concept satisfiability $\mathcal{ALCQI}$, where $\mathcal{ALCQI}$ is the well-known 
extension of standard $\mathcal{ALC}$ with qualified number restrictions and inverses. 
For the reduction, we define a translation from $\mathcal{ALCQP}(p,s)$-concepts
into $\mathcal{ALCQI}$-concepts that turns
higher-arity relations into binary ones, i.e., we define a suitable \emph{reification} procedure.
We translate $\mathcal{ALCQP}(p,s)$-concepts into $\mathcal{ALCQI}$-concepts using an operator $\mathbb{T}$.
First, for each atomic $\mathcal{ALCQP}(p,s)$-role symbol $R$, we define the fresh atomic concept symbol $L_R$.
Then, for each $n\geq 2$, we define the concept
$$Outdeg_n\ :=\ ({=}1 F_{1}.\top \sqcap \ldots \sqcap {=}1 F_{n}.\top)
\sqcap (\forall F_{1}. C_{dom} \sqcap \ldots \sqcap \forall F_{n}. C_{dom}),
$$
where $C_{dom}$ is a fresh concept name and each $F_i$ a fresh binary role name.
At those elements where $Outdeg_n$ holds, it
restricts the out-degree of each $F_i$ to one,
forcing $F_i$ to behave like a functional\footnote{Strictly speaking, $F_i$ will be a partial function.}  role
whose range consists of elements satisfying $C_{dom}$. 

Now, suppose we have fixed a context with an $n$-ary role symbol $R$ and a
permutation string $\pi$. Let $R^{\mathcal{I}}$ and $(R^{\pi})^{\mathcal{I}}$ be corresponding 
relations in an interpretation $\mathcal{I}$. For example, if $R^{\mathcal{I}} = \{(u_1,\dots , u_n)\}$
and $\pi = pp$,
then $(R^{pp})^{\mathcal{I}} = \{(u_{n-1}, u_n,u_1,\dots , u_{n-2})\}$. In this case, we notice, e.g., 
that the $n$th coordinate of (the tuple of) $R^{\mathcal{I}}$ gets sent to the second coordinate
place of (the tuple of) $(R^{pp})^{\mathcal{I}}$, and the coordinate $(n-1)$ gets sent to the first coordinate place.
We let $\pi_i \in \{1, \dots , n\}$ 
denote the coordinate place of $R^{\mathcal{I}}$
that gets sent to the $i$th coordinate place of $(R^{\pi})^{\mathcal{I}}$. For instance, in the above example 
where $R^{\mathcal{I}} = \{(u_1,\dots , u_n)\}$
and $\pi = pp$ and thus $(R^{\pi})^{\mathcal{I}} = \{(u_{n-1}, u_n,u_1,\dots , u_{n-2})\}$, we
have $\pi_1 = n-1$ and $\pi_2 = n$.

We are now ready to define a translation operator $\mathbb{T}$. Suppose we are translating a
concept $C_0$ of $\mathcal{ALCQP}(p,s)$ whose set of role symbols is $\textbf{S}$.
Then, for each $R\in \textbf{S}$, we let 
$$\chi_R\ :=\ L_R \sqcap \bigsqcap\limits_{ S \in \mathbf{R}\setminus \{R\}}\neg L_S,$$
that is, $\chi_R$ asserts that $L_R$ holds at the current point while all the
other concepts $L_S$ are false. Now, the translation $\mathbb{T}$ is
defined for the subconcepts of $C_0$ as follows.

\begin{enumerate}
\item $\mathbb{T}(\top) = \top$ and $\mathbb{T}(\bot) = \bot$
\item $\mathbb{T}(A) = A$ for a an atomic concept $A$
\item $\mathbb{T}(C_1 \sqcap C_2) = \mathbb{T}(C_1) \sqcap \mathbb{T}(C_2)$
\item $\mathbb{T}(\neg C) = \neg \mathbb{T}(C)$ 
%

%
%
%
%

%
%
\item $\mathbb{T}(\, {\geq} k R^{\pi}. (C_2,\ldots, C_{n})\, ) = \\$
   ${\geq}k F_{\pi_1}^{-1}.\big(\, \neg C_{dom} \sqcap \chi_R \sqcap  Outdeg_n  \sqcap 
 \exists F_{\pi_2} . \mathbb{T}(C_2) \sqcap 
\ldots \sqcap   \exists F_{\pi_n} .\mathbb{T}(C_{n})   \, \big)$

\end{enumerate}

The item 5 shows the key idea of how we translate higher-arity relations into  a 
setting with only binary ones.
These kinds of processes are usually
referred to as \textit{reification}. Also \emph{binarization}
would be suitable name.
In our construction, every tuple $(u_1,\ldots,u_{n}) \in R^\mathcal{I}$ (where $n\geq 2$)
becomes associated with a (so-called) \textit{lantern element} $l \in L_R^\mathcal{J}$ in a new
interpretation $\mathcal{J}$ with only binary relations.
For the lantern element $l$, we 
have $(l, u_i) \in F_i^\mathcal{J}$ for each $i \in \{1,\ldots, n\}$.
In other words, the point $l$ satisfying $L_R^\mathcal{J}$
represents (or reifies) an $R$-tuple whose elements 
are images of the roles $F_i$ (see Figure \ref{reified_fig}).

\begin{figure}[ht!]\label{figure1}
\centering
\begin{tikzpicture}[scale=0.8]
\tikzstyle{Farrow} = [thick,->,>=stealth]
\tikzstyle{blackdot} = [draw, circle, fill, black, inner sep=1.5pt]

\node [blackdot, name=1] at (0,0) {};
\node  at (0.8,0.1) {$l \in L_R^{\mathcal{J}}$};
\node  at (-1.15,-0.7) {$F_1$}; 
\node  at (0.01,-0.7) {$F_2$}; 
\node  at (0.9,-0.7) {$F_n$}; 
\node [blackdot, name=2] at (-1.5,-1.5) {};
\node at (-1.8,-1.5) {$u_1$};
\node [blackdot, name=3] at (-0.5,-1.5) {};
\node at (-0.8,-1.5) {$u_2$};
\node at (0.07,-1.5) {$\cdots$};
\node [blackdot, name=4] at (1,-1.5) {};
\node at (0.6,-1.5) {$u_n$};
\draw (-0.5 ,-1.5) ellipse (55pt and 10pt);
\node at (-2.7,-1.0) {$C_{dom}^{\mathcal{J}}$};

\draw [Farrow] (1) -- (2);
\draw [Farrow] (1) -- (3);
\draw [Farrow] (1) -- (4);

\end{tikzpicture}
\caption{A tuple $(u_1, \ldots, u_n) \in R^\mathcal{I}$ reified. The element $l$ does not
belong to $C_{dom}^{\mathcal{J}}$ but instead acts as a point encoding the tuple $(u_1, \ldots, u_n)$.
The points $u_1,\dots , u_n$ are in $C_{dom}^{\mathcal{J}}$.} 
\label{reified_fig}
\end{figure}
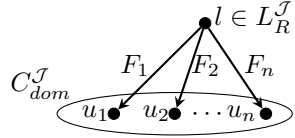

We use the rest of this section for
showing that an $\mathcal{ALCQP}(p,s)$-concept $C$ is 
satisfiable if and only if $C_{dom} \sqcap \mathbb{T}(C)$ is.
This reduces concept satisfiability of $\mathcal{ALCQP}(p,s)$ polynomially to the 
concept satisfiability problem of $\mathcal{ALCQI}$. We note that even binary roles of $\mathcal{ALCQP}(p,s)$
are reified in our translation.

We begin by proving the less laborious direction,
which is established in the following lemma.

\begin{lemma} \label{csat_empty_t_lemma1}
Let $C$ be an $\mathcal{ALCQP}(p,s)$-concept.
If $C$ is satisfiable, then 
so in $C_{dom} \sqcap \mathbb{T}(C)$.

\begin{proof}
Let $\mathcal{I}$ be a model of the $\mathcal{ALCP}(p,s)$-concept $C$, 
and let $m$ be the maximum arity of
roles in $C$.
We construct a model $\mathcal{J}$ for $C_{dom} \sqcap \mathbb{T}(C)$ from $\mathcal{I}$ as follows.
The domain $\Delta^{\mathcal{J}}$ will be the set $\Delta^{\mathcal{I}}\cup L$,
where $L$ contains, for each role name $R$ in $C$ and
each tuple $(u_1,\dots , u_n) \in R^{\mathcal{I}}$, a fresh 
element $u$. We may denote $u$ by $u_{R(u_1,\dots , u_n)}$. We note that $L$ stands for \emph{lantern elements}.
For each role $R$, we interpret the concept $L_{R}$ such that $(L_{R})^{\mathcal{J}}$
consists of exactly the 
points $u = u_{R(u_1,\dots , u_n)} \in L$.
We build each $F_i^{\mathcal{J}}$ in
the natural way (cf. Figure \ref{figure1}) such that
$$F_i^{\mathcal{J}}
: = \{ (u, u_i)\, |\, u = u_{R(u_1,\dots , u_i, \dots  , u_n)} \text{ for some } u_1,\dots , u_{i-1},u_{i+1},\dots , u_n\in \Delta^{\mathcal{I}} \}.$$
We define $C_{dom}^{\mathcal{J}} := \Delta^{\mathcal{I}}$ to denote the domain of $\mathcal{I}$, and
for each atomic concept $A$ occurring in $C$, we set $A^{\mathcal{J}} :=  A^\mathcal{I}$.

To conclude the proof, it suffices to establish by induction on the subconcepts $B$ of $C$ that for all elements $u\in \Delta^{\mathcal{I}} = (C_{dom})^{\mathcal{J}}$,
$$u \in B^{\mathcal{I}} \Leftrightarrow u\in 
(\mathbb{T}(B))^{\mathcal{J}}.$$
First, for each atomic concept $A$ in $C$, we
have $A^\mathcal{I} = A^\mathcal{J}$ by the definition of $\mathcal{J}$.
Assume then that for concepts $C_1,\ldots, C_{n}$ and for all $u \in \Delta^\mathcal{I}$, we have $u \in C_i^\mathcal{I}$ iff $u \in \mathbb{T}(C_i)^\mathcal{J}$. 
Let $t \in \Delta^\mathcal{I} = C_{dom}^\mathcal{J}$.
The Boolean cases are trivial:

\medskip

\noindent Conjunction: $t \in (C_1 \sqcap C_2)^\mathcal{I} \Leftrightarrow (t \in C_1^\mathcal{I}$  and   $t \in C_2^\mathcal{I}) \overset{\text{ ind. hypt.}}{\Leftrightarrow} 
 (t \in \mathbb{T}(C_1)^\mathcal{I}$ and  $t \in \mathbb{T}(C_2)^\mathcal{I})  \Leftrightarrow  
 t \in \mathbb{T}(C_1 \sqcap C_2)^\mathcal{J}.$

\medskip

\noindent Negation: $ t \in (\neg C_1 )^\mathcal{I} \Leftrightarrow t 
\not \in C_1^\mathcal{I} \overset{\text{ ind. hypt.}}{ \Leftrightarrow }
t \not \in \mathbb{T}(C_1)^\mathcal{J} \Leftrightarrow t \in \mathbb{T}(\neg C_1)^\mathcal{J}.$

\medskip

Finally, let us consider the case for cardinality restrictions.
Suppose that $t \in (\, {\geq} k R^{\pi} . (C_2,\ldots, C_{n})\, )^\mathcal{I}$.
This means that there exist at least $k$ tuples $\bar{u} = (u_1,\ldots, u_n) \in (R^{\pi})^\mathcal{I}$
such that $u_1 = t$ and $u_i \in C_i^\mathcal{I}$ for each $i\in \{2,\dots , n\}$. 
By the construction of $\mathcal{J}$, there exist at least $k$ corresponding lantern points in $L$ accessible via $(F_{\pi_1}^{-1})^{\mathcal{J}}$ from $t$, one lantern point for each tuple $\bar{u}$.
Every such lantern point satisfies $\neg C_{dom}\sqcap \chi_{R}$
and points via $F_{\pi_i}^{\mathcal{J}}$ (for each $i\in\{2,\dots , n\}$) to an
element in $\Delta^{\mathcal{I}}$ satisfying $\mathbb{T}(C_i)$ by the induction hypothesis.

Conversely, if $t$ satisfies $${\geq}k F_{\pi_1}^{-1}.\big(\, \neg C_{dom} \sqcap \chi_R \sqcap Outdeg_n  \sqcap 
 \exists F_{\pi_2} . \mathbb{T}(C_2) \sqcap 
\ldots \sqcap   \exists F_{\pi_n} .\mathbb{T}(C_{n})   \, \big),$$ there are at least $k$ lantern points satisfying
$$\neg C_{dom} \sqcap \chi_R \sqcap Outdeg_n  \sqcap 
 \exists F_{\pi_2} . \mathbb{T}(C_2) \sqcap 
\ldots \sqcap   \exists F_{\pi_n} .\mathbb{T}(C_{n})$$
and accessible from $t$ via $(F_{\pi_1}^{-1})^{\mathcal{J}}$. By our construction, there exist at least $k$
tuples $(u_1,\ldots, u_n) \in (R^{\pi})^\mathcal{I}$
such that $u_1 = t$ and $u_i \in C_i^\mathcal{I}$ for each $i\in \{2,\dots , n\}$.
\end{proof}
\end{lemma}

For the remaining direction from $\mathcal{ALCQI}$ to $\mathcal{ALCQP}(p,s)$, we define a variant of the standard unraveling (see, e.g., \cite{DBLP:books/cu/BlackburnRV01}) of models.
The unraveling procedure is needed for the following reason.
Let $\mathcal{I}$ be a model of $C_{dom} \sqcap  \mathbb{T}(C)$,
and let $l_1, l_2 \in \Delta^\mathcal{I}$ be
distinct lantern points for $R$, meaning that $l_1$ and $l_2\not= l_1$ satisfy
 $$\neg C_{dom} \sqcap \chi_R \sqcap Outdeg_n  \sqcap \exists F_{\pi_2} . \mathbb{T}(C_2) \sqcap 
\ldots \sqcap   \exists F_{\pi_n} .\mathbb{T}(C_{n}).$$
In an ideal situation, $l_1$ and $l_2$ represent distinct $n$-tuples,
but it is possible that they in fact encode the 
same tuple $(u_1,\dots , u_n)$, that is, $(l_1, u_i)$ and $(l_2, u_i)$ are in  $F_i^\mathcal{I}$ for
all $i \in \{1,\ldots, n \}$.
Working with an unraveling of $\mathcal{I}$ would fix this problem, but in the standard unraveling (with inverses), the in and out-degrees of points with respect to roles are not generally preserved.
For this reason, we define a suitable unraveling that preserves the in and out-degrees.

In the following, for a signature $\sigma$ with at most binary role
symbols, we call a $\sigma$-structure $\mathcal{I}$ \textit{tree-like} if 
$(\Delta^\mathcal{I}, E)$ is a tree, where $E$ is the
union of all the roles $R^\mathcal{I}$, $R \in \sigma$, and their inverses $(R^{-})^\mathcal{I}$.

\paragraph{Graded unraveling (g-unraveling).}
Let $\sigma$ be a
signature with at most binary relations.
Let $\mathcal{I}$ be a $\sigma$-interpretation and $r$
an element in the domain of $\mathcal{I}$.
The \textit{g-unraveling} of $\mathcal{I}$ from $r$ is the tree-like structure $\mathcal{I}^*$
defined as follows.

Let $\mathcal{E}_\sigma = \mathbf{R}_\sigma \cup \{ R^{-} \mid R \in \mathbf{R}_\sigma \}$, where $\mathbf{R}_\sigma$ consists of all role names in $\sigma$.
The domain of $\mathcal{I}^*$ consists of the singleton tuple $(r)$ and all finite tuples (walks) of the form  $(u_0,E_1,u_1,\ldots,E_n, u_n)$
such that the following conditions hold.
\begin{enumerate}
    \item 
$u_0 = r$.
\item
$u_i \in \Delta^\mathcal{I}$ 
and 
$E_i \in \mathcal{E}_\sigma$ for each $i\in \{0,\dots , n\}$.
\item
$(u_i,u_{i+1})\in E^\mathcal{I}_{i+1}$ for each $i\in \{0,\dots , n-1\}$.
\item
The following does not hold:
\begin{align*}
u_{i-1}= u_{i+1}\text{ and }E_i
\text{ is the inverse of }E_{i+1}.
\end{align*}
\end{enumerate}
Intuitively, traveling back and forth between two elements through a role $R$ and its inverse $R^-$ is not allowed, e.g.,
$uEvE^-u$ and $uE^-vEu$ are forbidden parts of walks.

%

%
For all elements $(r,\ldots, E_n, u_n) \in \Delta^\mathcal{I^*}$ (including $(r) \in \Delta^\mathcal{I^*}$) and all
atomic concepts $A \in \sigma$, we set $(r,\ldots,E_n, u_n) \in A^\mathcal{I^*}$ if $u_n \in A^\mathcal{I}$.
For all role symbols $R\in \mathbf{R}_\sigma$ and all elements $\overline{s},\overline{s}'\in \Delta^\mathcal{I^*}$, we set $(\overline{s},\overline{s}')
\in R^{\mathcal{I}^*}$ if one of the following conditions hold.
\begin{enumerate}
    \item 
    $\overline{s}'$ extends $\overline{s}$ by $(R,u)$ for some $u\in \Delta^{\mathcal{I}}$, i.e.,
    $\overline{s}' =  (\overline{s},R,u)$.
    \item
    $\overline{s}$ extends $\overline{s}'$ by $(R^-,u)$ for some $u\in \Delta^{\mathcal{I}}$, i.e.,
    $\overline{s} =  (\overline{s}',R^-,u)$.
\end{enumerate}  
We define the \emph{canonical map} $f:\Delta^{\mathcal{I}^*}\rightarrow \Delta^{\mathcal{I}}$ such
that $f(\overline{s})$ is the last element of $\overline{s}$ for all $\overline{s}\in \Delta^{\mathcal{I}^*}$. The map is clearly a surjection. It is easy to show by induction on concepts that each $\overline{s}\in \Delta^{\mathcal{I}^*}$ and $f(\overline{s})$ satisfy the same $\mathcal{ALCQI}$-concepts.
That is, the following holds.

\begin{lemma} \label{unravelling_lemma}
Let $\mathcal{I}$ be an interpretation and
$\mathcal{I}^*$ the g-unraveling of $\mathcal{I}$ from $r$. Let $f$ the related canonical map and $\overline{s}$ an element of $\mathcal{I}^*$.
Then we have $f(\overline{s})\in C^{\mathcal{I}}$ iff $\overline{s}
\in C^{\mathcal{I}^*}$ for
all $\mathcal{ALCQI}$-concepts $C$.
\end{lemma}

Now we are ready to cover the remaining direction concerning
the translation of $\mathcal{ACLQP}(p,s)$ to the binary realm.

\begin{lemma} \label{csat_empty_t_lemma2}
Let $C$ be a concept of $\mathcal{ALCQP}(p,s)$.
If $C_{dom}\sqcap\mathbb{T}(C)$ is satisfiable,
then so is $C$.
\begin{proof}

Assume first that $C_{dom} \sqcap \mathbb{T}(C)$ is satisfiable, so $r 
\in (C_{dom} \sqcap \mathbb{T}(C))^{\mathcal{I}'}$ for some
interpretation $\mathcal{I}'$
and some element $r$ in the domain of $\mathcal{I}'$. Let $\mathcal{I}$ be the g-unraveling of $\mathcal{I}'$ from $r$.
It follows from Lemma \ref{unravelling_lemma} 
that $(r)\in (C_{dom} \sqcap \mathbb{T}(C))^{\mathcal{I}}$.  Based on $\mathcal{I}$,
we will build an interpretation $\mathcal{J}$
such that $(r)^{\mathcal{J}} \in C^{\mathcal{J}}$.

We define the domain of $\mathcal{J}$ to be the 
set $C_{dom}^{\mathcal{I}}$, i.e., $\Delta^{\mathcal{J}} := C_{dom}^{\mathcal{I}}$.
Let $\sigma_p$ be the signature of $C$.
For every atomic concept $A \in \sigma_p$, we set $A^{\mathcal{J}} := (C_{dom} \sqcap A )^\mathcal{I}$.
Every $n$-ary role symbol $R  \in \sigma_p$ is interpreted such that $(u_1,\ldots,u_n)\in R^{\mathcal{J}}$ iff there exists a point $l \in \Delta^{\mathcal{I}}$ such that $(l,u_i) \in F_i^{\mathcal{I}}$ for each $i\in \{1, \dots , n\}$ and we
have $$l \in (\neg C_{dom}\sqcap \chi_R \sqcap \mathit{Outdeg}_n )^{\mathcal{I}}.$$

We then show, by induction on subconcepts of $C$, that the equivalence $$u \in B^\mathcal{J}\ \Leftrightarrow\ u \in (\mathbb{T}(B))^\mathcal{I}$$ holds for all $u\in \Delta^{\mathcal{J}} = C_{dom}^{\mathcal{I}}$
and all subconcepts $B$ of $C$.
For each atomic concept $A$, we have $A^\mathcal{J} = (C_{dom} \sqcap A)^\mathcal{I}$, so the equivalence is clear for $A$.
Assume now that for concepts $C_1, \ldots, C_{n}$ and for all $u \in \Delta^\mathcal{J}$, we have $u \in C_i^\mathcal{J}$ iff $u \in \mathbb{T}(C_i)^\mathcal{I}$. 
Let $t \in \Delta^\mathcal{J}$.
The Boolean cases of the induction are straightforward:

\medskip

\noindent Conjunction:
$t \in \mathbb{T}(C_1 \sqcap C_2)^\mathcal{I} \Leftrightarrow (t \in \mathbb{T}(C_1)^\mathcal{I}$ 
and $t \in \mathbb{T}(C_2)^\mathcal{I}) \overset{\text{ind. hypot.}}{\Leftrightarrow} (t \in C_1^\mathcal{J}$ and $t \in C_2^\mathcal{J}) \Leftrightarrow t \in (C_1 \sqcap C_2)^\mathcal{J}.$

\medskip

\noindent Negation:
$ t \in  \mathbb{T}(\neg C_1) ^\mathcal{I} \Leftrightarrow t \not \in \mathbb{T}(C_1)^\mathcal{I} \overset{\text{ind. hypot.}}{ \Leftrightarrow }
t \not \in C_1^\mathcal{J} \Leftrightarrow t \in (\neg C_1)^\mathcal{J}.$

\medskip

\noindent Assume then that 
$t \in (\, {\geq} k R^{\pi} . (C_2,\ldots, C_{n})\, )^\mathcal{J}$.
Therefore there exist at least $k$ tuples
$(t_2,\dots , t_n) \in (\Delta^{\mathcal{J}})^{n-1}$ 
%
%
%
such that $(t,t_2,\dots , t_n)\in (R^{\pi})^\mathcal{J}$
and $t_i \in C_i^{\mathcal{J}}$
for each $i\in \{2,\dots, n\}$. 
Thus, due to the way $\mathcal{J}$ is
defined from $\mathcal{I}$, there exist at least $k$ points $l$ in $\mathcal{I}$ such that $(l,t_i) \in F_{\pi_i}^{\mathcal{I}}$ for each $i\in \{2,\dots, n\}$ and $(l,t) \in F_{\pi_1}^{\mathcal{I}}$. Furthermore, we have $l \in (\neg C_{dom}\sqcap\, \chi_R\, \sqcap \mathit{Outdeg}_n )^{\mathcal{I}}.$ By the induction hypothesis, we have $t_i\in \mathbb{T}(C_i)^{\mathcal{I}}$ for each $i\in \{2,\dots , n\}$.
Therefore, we have 
$$l \in \big(\neg C_{dom} \sqcap \chi_R 
\sqcap Outdeg_n \sqcap \exists F_{\pi_2} . \mathbb{T}(C_2)  \sqcap 
\ldots \sqcap   \exists F_{\pi_n} .\mathbb{T}(C_n) \, \big)^\mathcal{I}.$$ 
Thus $t\in \mathbb{T}(\, {\geq} k R^{\pi} . (C_2,\ldots, C_{n})\, )^\mathcal{I}$.

For the converse, assume that we have $t\in \mathbb{T}(\, {\geq} k R^{\pi} .
(C_2,\ldots, C_{n})\, )^\mathcal{I}$. Then there exist at least $k$ points $l$ in $\mathcal{I}$ such that $(t,l)\in
(F_{\pi_1}^{-1})^{\mathcal{I}}$ and 
$$l \in \big(\neg C_{dom} \sqcap \chi_{R} 
\sqcap Outdeg_n \sqcap \exists F_{\pi_2} . \mathbb{T}(C_{\pi_2})  \sqcap 
\ldots \sqcap   \exists F_{\pi_n} .\mathbb{T}(C_{\pi_n}) \, \big)^\mathcal{I}.$$
By the definition of roles in $\mathcal{J}$, each of the
points $l$ defines a tuple $(t,u_2,\dots , u_n) \in (R^{\pi})^{\mathcal{J}}$
such that $u_i \in \mathbb{T}(C_i)^{\mathcal{I}}$ for all $i \in \{2,\dots , n\}$. Because $\mathcal{I}$ is an unraveling, we see that at least $k$ of the tuples $(t,u_2,\dots , u_n)$
are indeed different tuples.
By the induction hypothesis, we have $u_i \in C_i^{\mathcal{J}}$ for each $i \in \{2,\dots , n\}$. Thus $t \in (\, {\geq} k R^{\pi} . (C_2,\ldots, C_{n})\, )^\mathcal{J}$.
\end{proof}
\end{lemma}

We have now proven the following result. 
\begin{theorem}
Concept satisfiability for $\mathcal{ALCQP}(p,s)$ is \textsc{PSpace}-complete.

\begin{proof}
It follows from Lemmas \ref{csat_empty_t_lemma1} and \ref{csat_empty_t_lemma2} that 
concept satisfiability for the logic $\mathcal{ALCQP}(p,s)$ can be reduced in polynomial time
to the corresponding problem for $\mathcal{ALCQI}$, which is
known to be \textsc{PSpace}-complete (see \cite{DingHaWiDL2007} for the upper bound).  The lower
bound is immediate as $\mathcal{ALCQP}(p,s)$ contains $\mathcal{ALCQI}$ as a fragment essentially.
\end{proof}

\end{theorem}

As (a small selection of many) 
examples of particular systems of interest that extend $\mathcal{ALCQP}(p,s)$, we mention here the logics
$\mathcal{ALCQP}(p,s,\cup,\cap,\setminus)$ and $\mathcal{ALCQP}(p,s,I,\cup,\cap,\setminus)$.
As further interesting related system, we mention 
$\mathcal{ALCQP}(p,s,\neg,\cap)$,
$\mathcal{ALCQP}(p,s,I,\neg,\cap)$, $\mathcal{ALCQP}(e,p,s,\neg,\cap)$
and also $\mathcal{ALCQP}(e,p,s,I,\neg,\cap)$, which relates
rather closely to $\mathcal{DL}_{FU_1}$ from \cite{DBLP:conf/dlog/Kuusisto16}.
These systems extend nicely the basic polyadic Boolean modal logic $\mathcal{ALCP}(\neg,\cap)$.

The mentioned systems (and of course various other polyadic modal logics with different operators) relate
nicely to binary modal logics. This link concerns the complexities of the reasoning problems, but also 
many other issues, for example expressivity. For instance, obtaining expressivity characterizations is 
easy. Consider, for example, $\mathcal{ALCQP}(p,s,\cup,\cap,\neg)$. To characterize 
expressivity with formulae of modal nesting depth $k$ and counting up to the positive integer $p$, we
can define the following game. Firstly, for the duplicator to survive the zeroeth round, the two 
elements $w$ and $w'$ must satisfy the same propositions. Then, the game moves are as follows.

\begin{enumerate}
\item
The spoiler chooses $n\leq p$ tuples of length $m$ that 
originate from the current point $w$ in one model.
The tuples must satisfy the same role type (including permutations, but
repetitions of points can be ignored).
\item
The duplicator must respond with $n$ corresponding tuples originating 
from the current point of the other model. The tuples must match the role type of
the tuples of the spoiler (including permutations but without caring about repetitions).
\item
The spoiler chooses a point from one of the tuples.
\item
The duplicator must respond by a point from a chosen tuple in the 
other model (and that tuple must have the same coordinate index as the 
element chosen by the spoiler).
\end{enumerate}

These new points give the new position, and the same propositions must be satisfied.
The game ends when round $k$ finishes (unless the duplicator has lost even before that by
not being able to respond to the spoiler's moves).

\section{An algebra for \texorpdfstring{$\mathcal{ALC}$}{ALC} }

In this section we show how to define $\mathcal{ALC}$ (or standard multimodal 
logic) as a general relation algebra. This result is very easy to obtain, the purpose for proving it is simply to illustrate the uses of general relation algebras with a simple and concrete example. In general, directly algebraic approaches can have advantages over
ones based on polyadic modal logic. 
We first define some operators. The interested reader may compare these operators to the ones defined in \cite{alg2021} and \cite{ordered}.
Let us begin by defining the \emph{suffix intersection operator} $\dot\cap$
from \cite{alg2021}. 
Let $t_1$, $t_2$ be
terms and $t_1^{\mathcal{I}}$, $t_2^{\mathcal{I}}$ their interpretations. Suppose the arities of $t_1$
and $t_2$ are $k$ and $\ell$, respectively (note that the arity of a term $t$ is simply the arity of the relation $t^{\mathcal{I}}$ in an arbitrary interpretation $\mathcal{I}$.
Call $m = \mathit{max}(k,\ell)$.
Then the interpretation $(t_1\, \dot\cap\, t_2)^{\mathcal{I}}$ of the term $t_1\, \dot\cap\, t_2$ is 
\[\{\, (a_1,\dots ,a_m)\, |\, (a_{m - k+1},
\dots ,a_m) \in t_1^{\mathcal{I}}
\text{ and } (a_{m - \ell +1},
\dots ,a_m) \in t_2^{\mathcal{I}}\, \} . \]
Note above that if $k$ or $\ell$ is zero,
then $(a_{m+1}, a_m)$ denotes the empty tuple $\epsilon$.
Note that $\epsilon \in t^{\mathcal{I}}$ corresponds intuitively to $t$ being true.

Let us then define the \emph{one-dimensional projection} operator $\exists_1$ as follows. 
If the arity of $t^{\mathcal{I}}$ is at most one,
then $(\exists_1 t)^\mathcal{I} = t^{\mathcal{I}}$, and
otherwise $(\exists_1 t)^\mathcal{I}$ is
\[ \{\ a_1\ |\ 
     (a_1,\dots , a_k) \in    
     t^{\mathcal{I}}\text{ for some }a_2, \dots , a_k \text{ in }\mathcal{I}\ \}.\]

Let us also define the \emph{unary intersection} operator $\cap_1$ as follows.
\begin{enumerate}
\item
If at least one of $t_1^{\mathcal{I}}$
and $t_2^{\mathcal{I}}$ has arity at most one, then
$(t_1\, \cap_1\, t_2)^{\mathcal{I}}
=
\bigl(\exists_1(t_1\, \dot\cap\, t_2)\bigr)^{\mathcal{I}}$.
\item
Otherwise $(t_1\, \cap_1\, t_2)^{\mathcal{I}}
= \emptyset_1$, i.e., the unary empty relation.
\end{enumerate}

Finally, let us define the \emph{unary negation} operator $\neg_1$ such that $(\neg_1 t)^{\mathcal{I}}$ is
the complement of $t^{\mathcal{I}}$ 
if the arity of $t^{\mathcal{I}}$ is at most one, and otherwise $(\neg_1 t)^\mathcal{I}$ is $\emptyset_1$.

Now, for any $\mathrm{GRA}(f_i)_{i\in I}$, we
let $\mathrm{GRA}_k(f_i)_{i\in I}$ denote the 
subsystem of $\mathrm{GRA}(f_i)_{i\in I}$ where
every atomic relation symbol in every term
has arity at most $k$.

We then show that $\mathrm{GRA}_2(\neg_1, \cap_1)$ and 
$\mathcal{ALC}$ with roles are
equiexpressive. This means that $\mathrm{GRA}_2(\neg_1, \cap_1)$ can define precisely the concepts \emph{and roles}\footnote{We note that $\mathcal{ALC}$ of course has no other roles than atomic ones.} of $\mathcal{ALC}$, and vice versa, the unary (respectively, binary) relations definable by $\mathrm{GRA}_2(\neg_1, \cap_1)$ correspond to concepts (respectively, roles) of $\mathcal{ALC}$. 
We assume the underlying vocabulary contains $\bot$ and $\top$ as built-in unary relation symbols corresponding to the unary empty set $\emptyset_1$ and full domain.\footnote{This assumption is not necessary if the underlying vocabulary
contains at least one concept symbol.} Thus $\top$ and $\bot$ are
available as unary relations in both $\mathcal{ALC}$ and the algebra.
We also assume that the underlying vocabulary contains only unary and binary relation symbols, so
nullary relation symbols are excluded, as they are not present in standard $\mathcal{ALC}$.


\begin{proposition}
$\mathrm{GRA}_2(\neg_1, \cap_1)$ and
$\mathcal{ALC}$ with roles are
equiexpressive.
\end{proposition}

\begin{proof}
The translation $T$ from $\mathcal{ALC}$ 
to the algebra is as follows.

\begin{enumerate}
\item
$T(R) = R$ \text{ for an atomic role $R$}
    \item
    $T(A) = A$ \text{ for an atomic concept $A$}
    \item
    $T(\top) = \top$ and $T(\bot) = \bot$ 
    \item
    $T(\varphi\sqcap \psi)
    = T(\varphi)\, \cap_1\, T(\psi)$
    \item
    $T(\neg \varphi) = \neg_1 T(\varphi)$
    \item
    $T(\exists R. \varphi) = R \, \cap_1\, 
    T(\varphi)$
\end{enumerate}

We then consider the inverse translation. Note
that in $\mathrm{GRA}(\neg_1, \cap_1)$, it is
clear that the arity of $t^{\mathcal{I}}$
equals the arity of $t^{\mathcal{J}}$ for any
interpretations $\mathcal{I}$ and $\mathcal{J}$. (This essentially stems from both of the operators being \emph{arity regular}, i.e., operators such that the arity of an output relation is completely determined by the arities of the input relations). Therefore we can talk about the arity of a term $t$, and it is guaranteed to be equal to the arity of the corresponding relation $t^{\mathcal{I}}$ in any interpretation $\mathcal{I}$.
The translation from 
$\mathrm{GRA}_2(\neg_1, \cap_1)$ to
$\mathcal{ALC}$ with roles is as follows.
\begin{enumerate}
\item
$S(R) = R$ \text{for an atomic role $R$}
    \item
    $S(A) = A$ \text{for an atomic concept $A$}
    \item
    $S(\top) = \top$ and $S(\bot) = \bot$
    \item
    $S(\neg_1 t) = \begin{cases}\neg S(t) & \text{if the arity of $t$ is one}\\
    \bot & \text{otherwise}
    \end{cases}$
    \item
    $S(t_1\, \cap_1\, t_2)
    = \begin{cases}S(t_1)\, \sqcap\, S(t_2)&\text{ if $t_1$ and
                                                  $t_2$ are unary}\\
                                                  \bot &\text{ if $t_1$ and
                                                  $t_2$ are binary}\\
                                                  \exists R. S(t) & \text{ if one term $t\in \{t_1,t_2\}$ is 
                                                  unary and}\\
                                                  & \text{ the other term $R\in \{t_1,t_2\}$
                                                  binary}
    \end{cases}$
\end{enumerate}
\end{proof}

\section{Conclusions}

The approach to description logics via polyadic modal logic is flexible and natural, generalizing
the usual approach. Integrating general relation algebras into the picture makes this issue even 
more explicit. We promote the metaprinciple that almost all results for the binary case have a
corresponding result  in the polyadic framework (with the same complexities). This leaves a 
lot of ground for further work. 

Various fragments of $\mathrm{GRA}(e,p,s,I,\neg,J,\exists,\cup,\cap,\setminus,\dot{\cap},
{\cap}_1,\neg_1)$
can be conceived as relevant for research that relates to description logics (and other 
similar fields). And obviously there is a huge range of further entirely natural operators not
listed here.

\vspace{7mm}

\noindent
\textbf{Acknowledgements.} The authors were supported by
the Academy of Finland project \emph{Theory of computational logics}, 
grants 324435 and 328987.

\bibliography{polyALCIQ}

\appendix

\end{document}